\documentclass[10pt, doublecolumn]{IEEEtran}

\usepackage{cite,algorithm,algorithmic,amsbsy,amsmath,amssymb,epsfig,bbm,mathrsfs,bbm}  %,mathabx
\usepackage{multirow}
\usepackage{mathrsfs}
\usepackage{epstopdf}
\usepackage{bm}
\usepackage{verbatim}
\usepackage{amsthm}
\newtheorem{proposition}{Proposition}[section]

\newtheorem{theorem}{Theorem}[section]% theorem counter resets every \subsection
\newtheorem{lemma}{Lemma}[section]% theorem counter resets every \subsection

\usepackage{stfloats} 
\usepackage{url}
\usepackage{color}
\usepackage{acronym} 

\usepackage{ifpdf}
\usepackage{graphicx,cite}
 
\usepackage{color}
\usepackage{xcolor}
\usepackage{graphicx}

\usepackage{amsmath,amsfonts,amssymb,amsthm}

\usepackage{enumerate} 

\makeatletter
\def\bstctlcite{\@ifnextchar[{\@bstctlcite}{\@bstctlcite[@auxout]}}
\def\@bstctlcite[#1]#2{\@bsphack
  \@for\@citeb:=#2\do{%
    \edef\@citeb{\expandafter\@firstofone\@citeb}%
    \if@filesw\immediate\write\csname #1\endcsname{\string\citation{\@citeb}}\fi}%
  \@esphack}
\makeatother

\usepackage{amssymb}
  \newcommand*{\twoheadleftrightarrow}{%
    \twoheadleftarrow
    \mathrel{\mkern-15mu}%
    \twoheadrightarrow
  }

\begin{document}

\bstctlcite{BSTcontrol}

\begin{comment}
\acrodef{PPP}[PPP]{Poisson Point Process}
\acrodef{NPPP}[NPPP]{Non-homogeneous PPP}
\acrodef{PGFL}[PGFL]{Probability Generating Functional}
\acrodef{CDF}[CDF]{Cumulative Distribution Function}
\acrodef{PDF}[PDF]{Probability Distribution Function}
\acrodef{PMF}[PMF]{Probability Mass Function}
\acrodef{PCF}[PCF]{Pair Correlation Function}
\acrodef{RV}[RV]{Random Variable}
\acrodef{SIR}[SIR]{Signal-to-Interference Ratio}
\acrodef{i.i.d.}[i.i.d.]{independent and identically distributed}
\acrodef{w.r.t.}[w.r.t.]{with respect to}
\acrodef{MAC}[MAC]{Medium Access Control}
\acrodef{V2V}[V2V]{Vehicle-to-Vehicle}
\acrodef{1D}[1D]{one-dimensional}
\acrodef{2d}[2d]{two-dimensional}
\acrodef{VANET}[VANET]{Vehicular ad hoc network}

\acrodef{LT}[LT]{Laplace Transform}
\acrodef{CoV}[CoV]{Coefficient-of-Variation}
\end{comment}

\title{Two-Hop Connectivity to the Roadside in a VANET Under the Random Connection Model
% Counting Vehicles with Two-Hop Connectivity to a Roadside Unit Under the Random Connection Model
}

\author{Alexander P. Kartun-Giles, {\it{Member, IEEE}}, Konstantinos Koufos, Xiao Lu, {\it{Member, IEEE}}, and Dusit Niyato, {\it{Fellow, IEEE}}  % <-this % stops a space
\thanks{A. P. Kartun-Giles is with the School of Physics and Mathematical Sciences, Nanyang Technological University, Singapore. K.~Koufos is with the School of Mathematics, University of Bristol, U.K. X.~Lu is with Ericsson, Canada, and the Lassonde School of Engineering, York University, Canada. D.~Niyato is with the School of Computer Science and Engineering,  Nanyang Technological University, Singapore. This research is supported by the Ministry of Education, Singapore, under its AcRF Tier 1 grant MOE2018-T1-001-201 RG25/18.} \protect \\ 
%\thanks{This work was partially supported by the EPSRC grant number EP/N002458/1 for the project Spatially Embedded Networks. APKG was partially supported by the EPSRC project ``Random Walks on Random Geometric Networks'' (grant EP/N508767/1). All underlying data are provided in full within this paper.}
}

\maketitle

\begin{abstract}
In this paper, we compute the expected number of vehicles with at least one two-hop path to a fixed roadside unit (RSU) in a multi-hop, one-dimensional vehicular ad hoc network (VANET) where other cars can act as relays. The pairwise channels experience Rayleigh fading in the random connection model, and so exist, with a probability given by a function of the mutual distance between the cars, or between the cars and the RSU. We derive exact expressions for the expected number of cars with a two-hop connection to the RSU when the car density $\rho$ tends to zero and infinity, and determine its behaviour using an infinite oscillating power series in $\rho$, which is accurate for all regimes of traffic density. We also corroborate those findings with a realistic scenario, using snapshots of actual traffic data. Finally, a normal approximation is discussed for the probability mass function of the number of cars with a two-hop connection to the RSU.
\end{abstract}

\begin{IEEEkeywords}
Vehicular networks, end-to-end connectivity, random connection model, stochastic geometry, mobility traces. % ,  meta distribution.
\end{IEEEkeywords}

\section{Introduction}

The requirement for wireless communication technologies that sustain reliable connectivity between vehicles in smart motorways will become essential. A key element of the road infrastructure will be the roadside unit (RSU) with mounted stationary sensors and wireless connectivity. It is expected that the RSU will receive a variety of messages from vehicles not necessarily connected to each other, fuse the combined information and broadcast it to the vehicles~\cite{ETSIcam, ETSIcpm}. In this scenario, it is important that the broadcast reaches as many vehicles as possible. A practical solution to achieve this objective is to combine both Vehicle-to-Vehicle (V2V) and Vehicle-to-Infrastructure (V2I) communication, establishing connections between the vehicles and the RSU in a multi-hop fashion.

Motivated by this emerging scenario, the objective of this article is to develop our understanding of the RSU multi-hop communication range, i.e., how many vehicles are within $k$ hops of the RSU, as a function of the traffic density $\rho$ cars per unit length of road. In this case, we adopt a \textit{protocol-and-interference-free} approach, where the vehicles are modeled as the vertices of a one-dimensional (1D) soft random geometric graph (RGG)~\cite{penrose2016}. The reason is that we aim to identify the fundamental limits of multi-hop communication in VANETs, in particular, the simplest case, two-hop connectivity, between the vehicles and the RSU. We study the number of cars within $k=2$ hops to a single, fixed RSU, denoted as $N_2$, in a random geometric vehicular network built on a quasi-1D line (i.e. a single lane highway). The main technical contributions of our work are summarized as follows. 
\begin{itemize}
\item We derive the following infinite series expression for the expected value of $N_2$ in terms of $\rho$:
\[
  \mathbb{E} [N_2] = \sqrt{2 \pi}\rho\sum_{k=1}^\infty (-1)^{k-1}\frac{\big(\rho
    \sqrt{\pi/2} \big)^k}{k!\sqrt{k}}. 
\] 
\item We study the asymptotics of the $\mathbb{E} [N_{2}]$ as the car density approaches infinity and zero. Specifically, we show that the expected number of cars within two hops to a single, fixed RSU asymptotically grows as 
\[
\mathbb{E} [N_2] \simeq 2 \rho \sqrt{\log \big(\rho \sqrt{\pi/2} \big)} \, {\text{as}} \, \rho \to \infty.
\]
This is actually a counter-intuitive result, as one might have surmised that the quantity $\mathbb{E} [N_{2}]$ would grow proportionally to $\rho^2$. % as $\rho \to \infty$. 

%$\rho \to \infty$ and  $\rho \to 0$.
\item  We compare these analytic predictions to snapshots of real motorway traffic, showing how variations in the underlying stationary Poisson Point Process (PPP) model of the cars does not affect the results, and observe a normal approximation property for $N_2$. 
\end{itemize}

Two recent major works deal with the random connection model in the V2V and V2I settings~\cite{ng2011,zhang2012}. Ng et al. in~\cite{ng2011} present analytic formulas for the probability of two-hop connectivity, which gives the probability that all vehicles in a VANET are within two hops, given that the vehicles connect pairwise with probability $p$. This is developed for the general case of an arbitrary number of hops to the roadside by Zhang et al. in reference~\cite{zhang2012}. The study in \cite{Mao2010} calculates the probability that a node in a 1D PPP has a two-hop path to the RSU. However, the distribution of the number of nodes with a two-hop path is not treated therein, neither asymptotic equivalents for ultra-dense and sparse networks are derived, as we will do in this paper. In references \cite{kartungiles2018b} and~\cite{privault2019} the moments of the number of $k$-hop paths to the RSU are computed, which is a different metric than that studied in the present paper, where the moments of the number of nodes in the graph with a $k$-hop path to the RSU are calculated for $k=2$. Note that a node may have several two-hop paths to the RSU, e.g., more than one relay node to the RSU in the case of two-hop paths. 

%The rest of this paper is structured as follows.
In the rest of this paper we set up our VANET system model using the random connection model in Section \ref{sec:model}. In Section \ref{sec:analysis}, we state the main results of this paper. In Section \ref{sec:traces}, we corroborate our scaling laws with experimental evidence provided from road traffic data. We conclude in Section~\ref{sec:conclusion}.
\begin{figure}[!t] % [!t]
\centering
     \centering   
     \includegraphics[width=0.45 \textwidth]{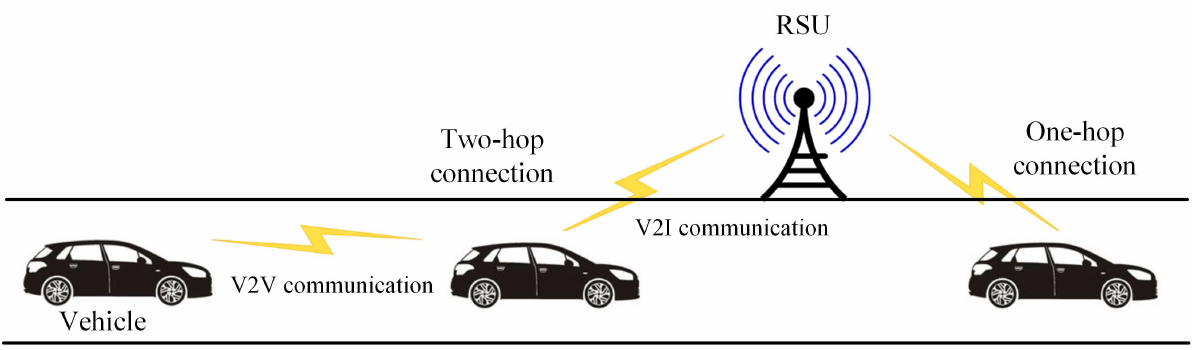}
     \centering
  \caption{An infrastructure-based VANET model with multi-hop connectivity.}
\label{fig:system}
\end{figure}

\section{System Model}
\label{sec:model}
We consider a \textit{random connection model}, which is an RGG $G_H (\mathcal{X})$ built on a stationary, homogeneous PPP $\mathcal{X}$ with flat intensity $\rho>0$ on an interval $V \subset \mathbb{R}$ centered at the origin. The nodes (or vertices) of the graph represent vehicles. In addition, we add a vertex at the location $u\in \mathbb{R}$ representing the RSU. See Fig.~\ref{fig:system}. The edges of the graph represent wireless communication links (known in physics as \textit{soft connectivity} \cite{giles2016}) with the connection function $H\!\left(r\right)$ modeling ``Rayleigh fading'' for  any  V2V  or  V2I  communication  link at distance $r$. The selection of Rayleigh fading is relevant to our system model; channel measurements have indicated that the narrowband small scale fading in V2V communication resembles Rayleigh for link distances larger than $50$ m~\cite{Cheng2007}. %~\cite[Table III and IV]{Cheng2007}.

The connection function $H(r)= \mathbb{P} ( x \leftrightarrow y ) \in [0,1]$ gives the probability that two nodes $x,y\in \mathbb{R}$ at distance $r=|x-y|$ in the graph are connected by an edge, resembling the long-term proportion of time that a Rayleigh fading channel at distance separation $r$ is in coverage~\cite{sklar1997-1}. It is a stretched exponential of the link distance $r$ and has been widely used to study connectivity in soft RGGs~\cite{penrose2016,giles2016}
\begin{equation}
 \label{eq:connH}
 H\!\left(r\right)\!=\!e^{-\beta r^\eta}, \, \beta\!>\!0.
\end{equation}

In the beginning of the next section, we will take $\beta = 1$ for brevity, and we will focus on the case with $\eta\!=\!2$, which captures free space path loss propagation. Our analysis stays the same for other values of $\beta\!>\!0$, while the assumption for $\eta=2$ will allow us to express multi-dimensional integrals in a closed-form and obtain important performance insights for multi-hop connectivity. It is also noted that to have a reliable two-hop path to the RSU, the values of the connection function for both hops (V2V and V2I) must be higher than or equal to a threshold. From this perspective, a {\it{decode-and-forward}} relay is essentially assumed here.

 \subsubsection*{Slivnyak-Mecke formula}
 We close this section with the statement of the Slivnyak-Mecke formula,  which will be used throughout the paper. It is an extension of Campbell's formula to sum over tuples of a point process. Often, and throughout this paper, we sum indicator functions over tuples of nodes of a point process,  and evaluate the typical value of the sum over all graphs.  We refer the reader to~\cite[Lemma 2.3]{penrose2016} for the following lemma.
\begin{lemma}[Slivnyak-Mecke Formula]
\label{lemma:mecke}
  Let $n\geq 1$.
  For any measurable real-valued function $g$ defined on the product of
  $V^n \times \mathcal{G}$, where $\mathcal{G}$ is the space of
  all graphs in $V$, the following relation holds:
 \begin{multline}\label{e:meckeformula}
   \mathbb{E} \!
   \left[
     \sum_{X_1 ,\dots,X_n \in \mathcal{X}}^{\neq} g(X_1 ,\ldots,X_n ,
     G_H (\mathcal{X}\setminus \{X_1,\ldots,X_n\}))
     \right]
   \\ =
   \rho^n % ( \mathbb{E} [\mathcal{X}] )^n
   \mathbb{E}  \! \left[
     \int_{V}
   \cdots\int_{V}
   g\left(x_1,\dots,x_n , G_H (\mathcal{X} ) \right) \mathrm{d}x_1
   \cdots \mathrm{d}x_n \right]\!, 
 \end{multline}
 where $\mathcal{X}$ is a PPP with intensity $\rho>0$ on $V$, and $\sum^{\not=}$ means the sum over all \textit{ordered} $t$-tuples of distinct points in $\mathcal{X}$. 
 \end{lemma}

\section{Results} \label{sec:analysis} 
This section presents the main technical results of our work. We start by introducing the preliminaries and then address the problem of determining the statistics of the number of vehicles on the road with a two-hop path connection to the RSU. 
\subsection{Preliminaries} 
\noindent 
In the sequel, we let
 \begin{equation}
 N_1(x,u) = \sum_{z \in \mathcal{X}}\mathbf{1}\left(x \leftrightarrow z \leftrightarrow u\right)
\end{equation}
 denote the number of distinct two-hop paths between vertices at $x$ and $u$ in $V$, and $\mathbf{1}(\cdot)$ being the indicator function. The following lemma is used throughout this article.
\begin{lemma}
\label{lemma:l1}
 The number $N_1(x,u)$ of two-hop paths between two given vertices $x, u \in V$ has a Poisson distribution.
\end{lemma}
\begin{proof}
    A proof of this result can be found in 
    \cite[Section 2]{privault2019}
    by computing the moments of $N_1(x,u)$ using the Slivnyak-Mecke formula and non-flat partitions.    
 \end{proof}
  
Next, we compute the mean number of two-hop paths which join two vertices of the graph, starting with the case of a finite interval of width $|V|<\infty$, which will later become the real line $\mathbb{R}$.
%A finite $\left|V\right|$ is applicable to urban street segments, while the infinite line is a good model for motorway traffic. Without any loss in generality, let us assume, herefater, that the RSU is located at $u$.
\begin{proposition}
  \label{p1}
         The probability of existence of at least one two-hop path between $x,u \in V$ is given by 
  \begin{equation}
  \label{e:ccl} 
\begin{array}{ccl}
  \mathbb{P}  ( x \twoheadleftrightarrow u ) \!\!\!&=&\!\!\! \displaystyle  1 -\exp\bigg(-
  \frac{\rho}{2} \sqrt{\frac{\pi}{2}} e^{-(x-u)^2/2} \,\,\, \times \\ \!\!\!\!\!\!\!\!\!\!\!\!& &\!\!\!\!\!\!\!\!\!\!\!\!    \left(\mathrm{erf}\left( \frac{|V|-(x-u)}{\sqrt{2}}\right) + \mathrm{erf}\left( \frac{|V|+x-u}{\sqrt{2}}\right)\right)
\bigg),
\end{array}
\end{equation}
where ${\mathrm{erf}}\left(x\right)\!=\! ( 2 /\sqrt{\pi} ) \int_0^x e^{-t^2}{\rm d}t$ is the error function.
\end{proposition}
 \begin{proof} % [Proof of Proposition \ref{p1}]
       The mean of $N_1(x,u)$ is given by the Campbell's theorem for point processes as follows 
   \begin{align}
   \mathbb{E}\left[ N_1(x,u)\right]  \nonumber % \label{e:mean0}
  &= \rho\int_{V}
  \mathbb{E} \left[\mathbf{1}\left(x \leftrightarrow z \leftrightarrow u\right) \right]\mathrm{d}z
  \\
  \nonumber
  &= \rho\int_{V}
  H(x,z)H(z,u)\mathrm{d}z \\
  \nonumber
  &= \rho\int_{-|V|/2}^{|V|/2}
  e^{- \left(x-z\right)^2 -(z-u)^2 }\mathrm{d}z
  \\
  \nonumber
  &=
   \frac{\rho }{2} \sqrt{\frac{\pi}{2}} e^{-(x-u)^2/2}
     \bigg (\mathrm{erf}\left( \frac{|V|-(x-u)}{\sqrt{2}}\right) \, + \\ 
& \,\,\,\,\, \mathrm{erf}\left( \frac{|V|+x-u}{\sqrt{2}}\right)\bigg). 
\end{align}

We conclude from Lemma~\ref{lemma:l1}, by noting that the probability that at least one two-hop path exists between $x$ and $u$ in $V$ is given by the complement of the void probability of a Poisson variate, as  
\begin{align}
 \mathbb{P} ( x \twoheadleftrightarrow u ) 
 &= 1- \mathbb{P} ( N_1(x,u) = 0 ) \nonumber \\
  &= 1-\exp\left(- \mathbb{E} [ N_1(x,u) ]\right). \nonumber 
\end{align}
 \end{proof}
   Proposition~\ref{c1} below, which treats the infinite length case $V = \mathbb{R}$, can be obtained either from Proposition~\ref{p1} by letting $|V|\to\infty$ in Eq.~\eqref{e:ccl}, or from Lemma~\ref{lemma:l1}.
\begin{proposition}
  \label{c1}
       The probability of existence of at least one two-hop path between $x$ and $u$ in $\mathbb{R}$ is given by 
       $$
 \mathbb{P} ( x \twoheadleftrightarrow u ) 
 = 1-\exp\left(-\rho \sqrt{\frac{\pi}{2}}e^{-(x-u)^2/2 }\right). 
$$ 
\end{proposition} 
\begin{proof}
 Here, we have 
 \begin{align}
   \nonumber % \label{e:mean1}
  \mathbb{E} [ N_1(x,u) ] 
  &= \rho\int_{-\infty}^\infty
  \mathbb{E} \left[\mathbf{1}\left(x \leftrightarrow z \leftrightarrow u\right) \right]\mathrm{d}z
  \\
  \nonumber
  %&= \rho\int_{-\infty}^\infty
  %H(x,z)H(z,u)\mathrm{d}z
  %\\
  %\nonumber
  &= \rho\int_{-\infty}^\infty
  e^{- \left(x-z\right)^2 -(z-u)^2 }\mathrm{d}z
  \\
  \nonumber
  &= \rho\sqrt{\frac{\pi}{2}} e^{-(x-u)^2/2}, 
\end{align}
 and we can conclude as in the proof of Proposition~\ref{p1}. % by Lemma~\ref{lemma:l1}
\end{proof}

\subsection{Main result}
Next, we assume that the RSU is located at $u$ and apply Proposition~\ref{c1} to analyze the mean number of vertices with two-hop connectivity to the RSU. 
 \begin{theorem}
  \label{t:expectation}
  Let $N_2$ be the number of vertices with two-hop connectivity to the RSU in a specific realisation of $G_H (\mathcal{X})$. %defined in Section \ref{sec:model}.
  \begin{enumerate}[i)] 
  \item The mean number of vertices with at least one two-hop path to the RSU is given by 
    \begin{align}      
      \mathbb{E} [N_2 ]
      % &= \rho  \int_{-\infty}^\infty
  % \big( 1-\exp\big(-\rho \sqrt{\pi / 2}e^{-(x-u)^2/2 }\big)\big)
  % \mathrm{d}x \nonumber \\
  &=   \rho\sqrt{2\pi}\sum_{k=1}^{\infty}(-1)^{k-1}\frac{\big(\rho \sqrt{\pi/2} \big)^k}{k!\sqrt{k}}. \label{e:a0}
\end{align}
\item In addition, we have the exact equivalents 
  \begin{align}
    \label{e:a1}
  \mathbb{E} [N_2 ] & \simeq  2 \rho \sqrt{ 2\ln\rho + 2\ln\sqrt{\pi/ 2}  }, \, {\text{as}} \,  \rho \to \infty.  \\ 
    \label{e:a2}
  \mathbb{E} [N_2 ] & \simeq \pi \rho^2, \, {\text{as}} \, \rho \to 0.  
\end{align}
  \end{enumerate}
\end{theorem}
\begin{proof} % [Proof of Theorem \ref{t:expectation}]
  \noindent
  $(i)$ We represent the number $N_2$ as the sum  % of vertices in $G_H (\mathcal{X})$ which  have graph distance of two hops to the RSU 
\[ 
  N_2 = \sum_{x \in \mathcal{X}}\mathbf{1}\left(x \twoheadleftrightarrow u\right).
\]

Next, we find via Lemma~\ref{lemma:mecke} and Proposition~\ref{c1}, that the expected value of this number is 
  \begin{eqnarray}
  \nonumber % \label{e:indicator}
  \mathbb{E} [N_2] & = & 
  \mathbb{E} \left[\sum_{x \in \mathcal{X}}\mathbf{1}\left(x \twoheadleftrightarrow u\right)\right]
  \\
  \nonumber
  & = & \rho\int_{-\infty}^\infty
   \mathbb{E} \left[\mathbf{1}\left(x \twoheadleftrightarrow u\right)\right] \mathrm{d}x
  \\
  \nonumber 
  \label{eq:1c}
  & = & \rho\int_{-\infty}^\infty
   \mathbb{P} ( x \twoheadleftrightarrow u ) \mathrm{d}x
  \\
\nonumber % \label{e:s1}
\label{eq:1d}
  & = & 
  \rho  \int_{-\infty}^\infty
  \big( 1-\exp\big(-\rho \sqrt{\pi / 2}e^{-x^2/2 }\big)\big)\mathrm{d}x.
  %\\ \nonumber 
  %& = & 
%\rho f ( \rho \sqrt{\pi / 2} ),
\end{eqnarray} 
Let us define $\alpha=\rho \sqrt{\pi / 2}$ and % where we have defined  
 \[
   f (\alpha ) = \int_{-\infty}^\infty
   \big( 1 - \exp \big( - \alpha e^{-x^2/2} \big) \big)\mathrm{d}x.
 \]

The integral in the expression of $f\!\left(\alpha\right)$ above can be expanded in a Taylor series about $\rho=0$ and integrated term by term to obtain the infinite oscillating sum % of Eq.~\eqref{e:a0},
\[
\begin{array}{ccl}
  f(\alpha ) &=& \displaystyle - \int_{-\infty}^\infty
   \sum\limits_{k=1}^{\infty}\frac{(-\alpha )^k}{k!}
  e^{-kx^2 /2}\mathrm{d}x \\ 
  &=& \displaystyle - \sqrt{2\pi}\sum\limits_{k=1}^{\infty} \frac{(-\alpha )^k}{k!\sqrt{k}},
\end{array}
\]
which proves Eq.~\eqref{e:a0}.

\noindent
$(ii)$ Next, we derive the asymptotics in \eqref{e:a1}. By doing a change of variable first $v=\frac{x}{\sqrt{2\log \alpha}}$ and then by splitting the integral, we have 
\begin{align}
  \nonumber
  f\left(\alpha\right) & = 
  \sqrt{2 \log \alpha }\int_{-\infty}^\infty 
 \big( 1 - \exp\big({-\alpha^{1-v^2}} \big)\big) \mathrm{d}v 
 \\
 & =  2 \sqrt{2 \log \alpha }\int_0^1 
 \big( 1 - \exp\big({-\alpha^{1-v^2}} \big)\big) \mathrm{d}v  \,\,\, + \nonumber \\ 
& \,   2 \sqrt{2 \log \alpha }\int_1^\infty
   \big(1 - \exp\big({-\alpha^{1-v^2}} \big)\big) \mathrm{d}v, 
 \, \alpha \geq 1. \label{e:f(a)0}
\end{align}

Using the bound  $0 \leq 1- e^{-\alpha^{1-v^2}} \leq e \alpha^{1-v^2}$, which is valid as $0 \leq \alpha^{1-v^2} \leq 1$ when $v \geq 1$
 and $\alpha \geq 1$, the second integral in Eq.~\eqref{e:f(a)0} can be bounded as % we have 
\[
\begin{array}{ccl}
 0 \!\!\!&\leq&\!\!\! \displaystyle \int_1^\infty
  \big( 1 - e^{-\alpha^{1-v^2}} \big) \mathrm{d}v
  \leq 
  \alpha
  e \int_1^\infty \alpha^{-v^2} 
 \mathrm{d}v \nonumber \\ 
  \!\!\!&\leq &\!\!\! \displaystyle \alpha
  e \int_1^\infty
 v e^{-v^2 \log \alpha } \mathrm{d}v \nonumber = 
 \frac{e}{2 \log \alpha },
 \quad \alpha \geq 1.
\end{array} 
\] 

Hence, as $\alpha\to\infty$ the second integral in Eq.~\eqref{e:f(a)0} becomes equal to zero. Furthermore, by the limit $\lim_{\alpha \to \infty} \alpha^{1-v^2}= +\infty$, $v\in [0,1)$, and dominated convergence,the first integral in Eq.~\eqref{e:f(a)0}  satisfies 
\[
\lim_{\alpha \to \infty}
 \int_0^1
 \big(1 - e^{-\alpha^{1-v^2}} \big) 
 \mathrm{d}v = 1.
 \]
 
%% From the inequalities 
%% \[
%%    \begin{array}{ccl}
%%  \displaystyle \int_0^1 \!\!\! 
%%   \big(1 - \exp\big({-\alpha^{1-v^2}} \big)\big) \mathrm{d}v  \!\!\!&\leq&\!\!\!  \displaystyle \int_0^\infty
%%    \!\!\! \big(1 - \exp\big({-\alpha^{1-v^2}} \big)\big) \mathrm{d}v \\ \!\!\!&\leq&\!\!\! \displaystyle \int_0^1 
%%    \!\!\! \big(1 - \exp\big({-\alpha^{1-v^2}} \big)\big) \mathrm{d}v   \, + \\ & & \displaystyle \frac{e}{2 \log \alpha}, 
%%    \end{array}
%% \]
We can therefore conclude that $ f(\alpha ) \simeq 2 \sqrt{2\log \alpha } \, {\text{as}} \, \alpha \to \infty$. The equivalent in Eq.~\eqref{e:a2} as $\alpha$ tends to zero follows by truncation of Eq.~\eqref{e:a0} at the first term. 
\end{proof}

For completeness, we also state the generalization of the Theorem~\ref{t:expectation} for arbitrary $\beta>0$ without proof. 
\begin{proposition}
For the connection function $H\!\left(r\right)=e^{-\beta r^2}$ with $\beta>0$, the mean number of vertices with at least one two-hop connection to the RSU is 
\[
\mathbb{E}\!\left[N_2\right] = \rho \sqrt{\frac{2\pi}{\beta}}\sum\limits_{k=1}^\infty (-1)^{k-1}\frac{\left(\rho\sqrt{\pi/\left(2\beta\right)}\right)^k}{k!\sqrt{k}}.
\]
In addition, we have the following exact asymptotics
\begin{align}
\mathbb{E}\!\left[N_2\right] & \simeq 2\rho\sqrt{\left(2/\beta\right)\ln\rho + \left(2/\beta\right)\ln\sqrt{\pi/\left(2\beta\right)}}, \, \rho \to \infty. \nonumber \\ 
\mathbb{E}\!\left[N_2\right] & \simeq  \frac{\pi \rho^2}{\beta}, \, \rho \to 0. \nonumber 
\end{align}
\end{proposition}

Theorem~\ref{t:expectation} is validated against simulations in Fig.~\ref{fig:realization2}. We see that the asymptotic expression in Eq.~\eqref{e:a1} remains very accurate also for realistic values of $\rho$, while the approximation for $\rho\rightarrow 0$ is accurate only for very sparse networks. In Fig.~\ref{fig:realization}, we illustrate that the probability of a two-hop path, see Eq.~\eqref{e:ccl}, for a high density $\rho=500$ sharply decreases beyond some distance from the RSU. Finally, in Fig.~\ref{fig:grid0} we observe that even for relatively low densities, a Gaussian distribution well approximates the probability mass function (PMF) of $N_2$, i.e., after rescaling 
\begin{equation}\label{e:p}
  \mathbb{P} \left(\frac{N_2 - \mathbb{E} [N_2]}{\sqrt{\mathrm{Var}[N_2]}} \leq x \right) \simeq \Phi(x), 
  \qquad x\in \mathbb{R}, 
\end{equation}
as $\rho\to\infty$, with $\Phi$ the cumulative distribution function of the standard normal distribution. Observing normal approximation is common in the case of independent trials. Here, the events ``a node $x$ has a two-hop path to the RSU'', for each $x \in \mathcal{X}$, are dependent, nevertheless, the dependency is sufficiently short range to lead to a Gaussian distribution in the dense limit. A proof for the Gaussian convergence is not a triviality and is left for future work. 
%% where the mean degree $\mu_d$ of $G_H (\mathcal{X})$, given by 
%% \begin{equation}
%% \label{e:p0}
%% \begin{array}{ccl}
%% \mu_d &=& \displaystyle \rho\int_{-\infty}^\infty
%% \mathbb{P} ( x \leftrightarrow 0 ) \mathrm{d}x \\ 
%% &=& \displaystyle 
%% \rho\int_{-\infty}^\infty H(r){\rm d}r \\ 
%% &=& \displaystyle \rho\int_{-\infty}^\infty e^{- r^2}{\rm d}r \\ 
%% &=& \displaystyle \rho \sqrt{\pi}, 
%% \end{array}
%% \end{equation} 
%% goes to infinity as $\rho \to \infty$. 
\begin{figure}[!t]
   \centering
   \includegraphics[width=3in]{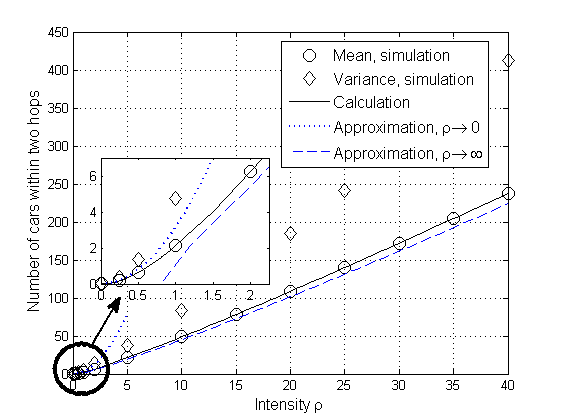}
   \caption{The mean number of vehicles with a two-hop path connection to the RSU calculated using Eq.~\eqref{e:a0} is compared against $20\,000$ Monte Carlo simulation runs over a line segment $\mathcal{X} \subset [-10,10]$ with the connection function $H(r)=\exp(-r^2)$. The simulated variance is depicted too.} % $V\!=\!\sqrt{2\log\! (- \rho \sqrt{\pi/2} / \log\! (1-\epsilon  ) )}$, and the constant $\epsilon\!\ll\! 1$ is positive, $\epsilon\!=\! 10^{-8}$ in our simulations. Note that the vehicles at distances larger than $V$ from the RSU placed at the origin have a two-hop connection with probability less than $\epsilon$ and thus they do not contribute essentially to the value of $N_2$.
   % For large values of $\rho$ many terms in Eq.~\eqref{e:a0} contribute to the sum, and the oscillating series may start to diverge due to precision errors. For $\rho\!>\!20$ we have truncated the series in Eq.~\eqref{e:a0} at the $250$-th term, and we have used high precision arithmetics in Mathematica at the $50$-th digit.
 \label{fig:realization2}
 \end{figure}
\begin{figure}[!t]
   \centering   
   \includegraphics[width=3in]{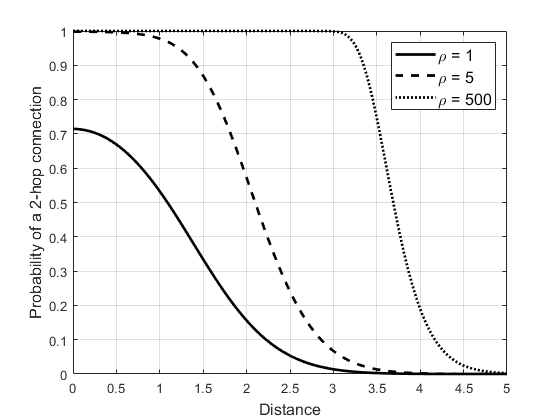}
   \caption{The probability of existence of a two-hop path as a function of the distance to the RSU.}
 \label{fig:realization}
\end{figure}
\begin{figure}[!t] 
  \centering
   \centering   
       \includegraphics[width=0.45 \textwidth]{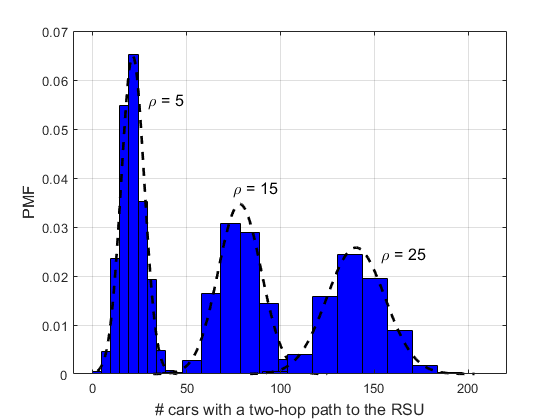}
       \caption{Counts of nodes with two-hop connectivity to the RSU for various densities $\rho$ of cars. See the caption of Fig.~\ref{fig:realization2} for the rest of the parameter settings.}
       \label{fig:grid0}
\end{figure}

\section{Validation with synthetic traces}\label{sec:traces}
Next, we assess the sensitivity of our model against the Poissonian assumption for the distribution of vehicles along a motorway. Using synthetic mobility traces~\cite{Gramaglia2014,Gramaglia2016}, we simulate the distribution of the number of vehicles with at least one two-hop link to the RSU, and compare this distribution (also referred to as the empirical distribution) to a Gaussian with mean equal to the parameter calculated in Eq.~\eqref{e:a0}. Overall, the spatial distribution of vehicles, at a snapshot of time, along a motorway, is not exactly Poisson. Therefore it is important we quantify the sensitivity of the mean calculated in Eq.~\eqref{e:a0} to small perturbations to the deployment model.

Thanks to~\cite{Gramaglia2014,Gramaglia2016}, $1800$ consecutive snapshots of road traffic along a $10$-km three-lane motorway are publicly available. The associated data files contain the horizontal location and occupied lane for all the vehicles over the snapshots. The time granularity is set to one second, so the total simulation time is half an hour, see~\cite[Section III-A]{Gramaglia2014}. In our simulations, we will actually use the last $1200$ out of the total $1800$ snapshots, in order to allow the first vehicles entering the roadway to reach the exit. This is to ensure that for all considered snapshots there are enough vehicles along the highway, hence, we can accurately simulate the statistics of headway distances. Since the considered snapshots cover a time interval of $20$ mins, there are slight variations in the intensity of vehicles during the simulation time, see the solid line in the inset of Fig.~\ref{fig:N2BusyAllsnap}.

We first generate the empirical distribution of headway distances for each snapshot. To do that we project all vehicles onto a single line, which apparently, does not introduce much error, because the communication range is expected to be much larger than the inter-lane distance separation. In each simulation run, we use inversion sampling of the empirical CDF to cover a line segment of $10$ km with vehicles. The set of vehicles is denoted by $\mathcal{Y}$, and the origin, where the RSU is located, is taken in the middle of the road segment, i.e., at $5$ km. Next, for each vehicle $y\!\in\!\mathcal{Y}$, we generate a random number distributed uniformly in $\left[0,1\right]$. After comparing it with the value of the connection function $e^{-\beta r^2}$, where $r$ is the distance between the vehicle and the RSU, we can identify the set of vehicles $\mathcal{Y}_1$, which have a single-hop communication link to the RSU. Then, we search over the vehicles $y\!\in\!\mathcal{Y}\backslash \mathcal{Y}_1$ and separate those with a single-hop connection to at least one of the vehicles in $\mathcal{Y}_1$. They become the elements of the set $\mathcal{Y}_2$, which consists of the vehicles with a two-hop path to the RSU. To complete the set $\mathcal{Y}_2$, we need to add to it the elements of the set $\mathcal{Y}_1$ which also have a two-hop connection to the RSU. Finally, the cardinality of the set $\mathcal{Y}_2$ is computed for each simulation run.

In Fig.~\ref{fig:N2BusyAllsnap}, it is illustrated that Eq.~\eqref{e:a0} yields a very good estimate of the simulated mean number of vehicles with a two-hop path to the RSU. Also, the quality of the Gaussian fit is good, even though the actual deployment of vehicles does not follow a PPP and the density of vehicles varies with time.   
\begin{figure}[!t] % [!t]
 \centering
\includegraphics[width=3in]{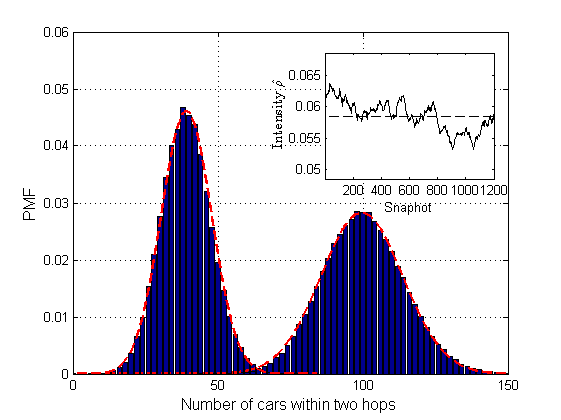}
\caption{The empirical PMF (blue bars), of the number of vehicles with a two-hop connection to the RSU averaged over $1200$ snapshots. For each snapshot we generated $100$ independent spatial configurations of vehicles by sampling its empirical CDF of inter-vehicle distances. Two values for the parameter $\beta$ are considered: $\beta\!=\! 5\times 10^{-5}$ and $\beta\!=\! 10^{-5}$. The simulated mean values for the number of vehicles with a two-hop path to the RSU is $39.51$ and $100.60$, respectively. The estimated intensity of vehicles $\hat{\rho}$ for each snapshot is depicted in the inset (solid line). The mean intensity of vehicles averaged over all considered snapshots is equal to $\hat{\rho}=0.0585\, {\text{m}}^{-1}$, see the dashed line in the inset. After substituting $\hat{\rho}=0.0585$ into Eq.~\eqref{e:a0}, we obtain the following values for the expected number of vehicles with a two-hop connection to the RSU: $39.12$ vehicles for $\beta\!=\! 5\times 10^{-5}$ and $99.60$ vehicles for $\beta\!=\! 10^{-5}$. The red dashed line corresponds to the Gaussian approximation for the distribution of $N_2$ using a mean equal to $\mathbb{E} [N_2]$ calculated from Eq.~\eqref{e:a0} and variance obtained by the simulations.}
\label{fig:N2BusyAllsnap}
\end{figure}
  
\section{Conclusions} \label{sec:conclusion}
Motivated by V2I communications in a VANET topology, we investigate the notion of $k$-hop connectivity to a fixed roadside unit (RSU), specifically, in the random connection model on a one-dimensional Poisson point process (PPP) with connection function $H(r) = \exp(-\beta r^2)$, and with $k=2$. We calculate the typical number of cars with a two-hop path connection to the RSU, denoted by $\mathbb{E} [N_2]$. This is then expanded in a power series to provide an infinite oscillating series expression for $\mathbb{E}[N_2]$, which, truncated to the first term, shows that the two-hop connectivity grows quadratically with traffic density, when the traffic is low. On the other hand, we have shown that $\mathbb{E}[N_2] = O(\rho \sqrt{\log \rho})$ as $\rho\to\infty$. Finally, we present the comparison with both Monte Carlo simulations, corroborating these formulas, and with a random model based on synthetic mobility traces. The simulations confirm that the typical number of vehicles in the  two-hop range of an RSU is an accurate representation of an interference-free network of moving cars, at a random instance of time, and further, that a Gaussian central limit theorem is also present. We have also identified the following two non-trivial problems as promising directions for future work: (i) Calculate the distribution of the number of $k$-hop paths between a node at $x$ and the RSU, which will enable us to compute the expected number of vehicles $\mathbb{E}[N_k]$ with a $k$-hop path to the RSU. (ii) Calculate higher moments of the number of cars with a two-hop path to the RSU, and use them to prove the Gaussian convergence. % of $N_2$ as $\rho\rightarrow\infty$.

%\section{Acknowledgements}

\footnotesize

\def\cprime{$'$} \def\polhk#1{\setbox0=\hbox{#1}{\ooalign{\hidewidth
  \lower1.5ex\hbox{`}\hidewidth\crcr\unhbox0}}}
  \def\polhk#1{\setbox0=\hbox{#1}{\ooalign{\hidewidth
  \lower1.5ex\hbox{`}\hidewidth\crcr\unhbox0}}} \def\cprime{$'$}

\end{document}